\documentclass[12pt,letterpaper]{article}
\usepackage{amsmath,amsfonts,amsthm,amssymb,stmaryrd,bm,cite} 
\usepackage{relsize,tabls}

\theoremstyle{plain}

\numberwithin{equation}{section}
\newtheorem{thm}{Theorem}[section]

\newcommand{\complex}{{\mathbb C}}
\newcommand{\real}{{\mathbb R}}
\newcommand{\lscript}{{\mathcal L}}
\newcommand{\pscript}{{\mathcal P}}
\newcommand{\rmre}{\mathrm{Re}}
\newcommand{\rmtr}{\mathrm{tr}}
\newcommand{\xhat}{\widehat{x}}
\newcommand{\yhat}{\widehat{y}}
\newcommand{\atilde}{\widetilde{a}}
\newcommand{\cbar}{\overline{c}}
\newcommand{\punderbar}{\underline{p}}
\newcommand{\qunderbar}{\underline{q}}
\newcommand{\junderbartilde}{\widetilde{\underline{j}}}
\newcommand{\sigmaunderbar}{\underline{\sigma}}
\newcommand{\gammaunderbar}{\underline{\gamma}}
\newcommand{\cappunderbar}{\underline{P}}
\newcommand{\ctimes}{\mathrel{\mathlarger\cdot}}

\newcommand{\ab}[1]{\left|#1\right|}
\newcommand{\brac}[1]{\left\{#1\right\}}
\newcommand{\paren}[1]{\left(#1\right)}
\newcommand{\sqbrac}[1]{\left[#1\right]}
\newcommand{\floorbrac}[1]{\lfloor#1\rfloor}
\newcommand{\elbows}[1]{{\left\langle#1\right\rangle}}

\begin{document}

\title{INFLATION AND DIRAC\\IN THE CAUSAL SET APPROACH\\
TO DISCRETE QUANTUM GRAVITY
}
\author{S. Gudder\\ Department of Mathematics\\
University of Denver\\ Denver, Colorado 80208, U.S.A.\\
sgudder@du.edu
}
\date{}
\maketitle

\begin{abstract}
In this approach to discrete quantum gravity the basic structural element is a covariant causal set ($c$-causet). The geometry of a $c$-causet is described by a shell-sequence that determines the discrete gravity of a universe. In this growth model, universes evolve in discrete time by adding new vertices to their generating $c$-causet. We first describe an inflationary period that is common to all universes. After this very brief cycle, the model enters a multiverse period in which the system diverges in various ways forming paths of $c$-causets. At the beginning of the multiverse period, the structure of a four-dimensional discrete manifold emerges and quantum mechanics enters the picture. A natural Hilbert space is defined and a discrete, free Dirac operator is introduced. We determine the eigenvalues and eigenvectors of this operator. Finally, we propose values for coupling constants that determine multiverse probabilities. These probabilities predict the dominance of pulsating universes.
\end{abstract}

\section{Causal Sets}  
This section presents a brief review of causal sets (causets). For more background and details, we refer the reader to \cite{hen09,sor03,sur11}. We call a finite partially ordered set $(x,<)$ a \textit{causet} and interpret $a<b$ in $x$ to mean that $b$ is in the causal future of $a$
\cite{hen09,sor03,sur11}. A \textit{labeling} for a causet $x$ of cardinality $\ab{x}$ is a bijection
\begin{equation*}
\ell\colon x\to\brac{1,2,\ldots ,\ab{x}}
\end{equation*}
such that $a,b\in x$ with $a<b$ implies $\ell (a)<\ell (b)$. A labeling for $x$ can be thought of as a ``birth order'' for the vertices of $x$. Two labeled causets $x,y$ are \textit{isomorphic} if there is a bijection $\phi\colon x\to y$ such that $a<b$ in $x$ if and only if $\phi (a)<\phi (b)$ in $y$ and
$\ell\sqbrac{\phi (a)}=\ell (a)$ for all $a\in x$. A causet is \textit{covariant} if it has a unique labeling (up to isomorphism) and we call a covariant causet a $c$-\textit{causet} \cite{gud13,gud141,gud142,gud15}. Denote the set of $c$-causets with cardinality $n$ by $\pscript '_n$ and the set of all
$c$-causets by $\pscript '$. It is shown in \cite{gud13,gud141} that the cardinality $\ab{\pscript '_n}=2^{n-1}$, $n=1,2,\ldots\,$. Two vertices
$a,b\in x$ are \textit{comparable} if $a<b$ or $b<a$. The \textit{height} $h(a)$ of $a\in x$ is the cardinality, minus one, of a longest path in $x$ that ends with $a$. It is shown in \cite{gud141} that a causet $x$ is covariant if and only if $a,b\in x$ are comparable whenever $h(a)\ne h(b)$.

If $x\in\pscript '$ we call the sets
\begin{equation*}
S_j(x)=\brac{a\in x\colon h(a)=j},\quad j=0,1,2,\ldots
\end{equation*}
\textit{shells} and the sequence of integers $s_j(x)=\ab{S_j(x)}$, $j=0,1,\ldots$, is the \textit{shell sequence} for $x$. A $c$-causet is uniquely determined by its shell sequence and we think of $\brac{s_i(x)}$ as describing the ``shape'' or geometry of $x$. In this model, we view a $c$-causet as a framework or scaffolding for a possible universe at a fixed time. The vertices of $x$ represent small cells that can be empty or occupied by a particle. We shall later consider a growth model in which universes evolve in time by adding new vertices.

Let $x=\brac{a_1,a_2,\ldots ,a_n}\in\pscript '_n$, where the subscript $i$ of $a_i$ is the vertex label, $\ell (a_i)=i$. A \textit{path} is a sequence
\begin{equation*}
\gamma =a_{i_1}a_{i_2}\cdots a_{i_m}
\end{equation*}
in $x$ starting at $a_{i_1}$ and moving along successive shells until $a_{i_m}$ is reached. We define the \textit{length} of $\gamma$ by
\cite{gud13}
\begin{equation*}
\lscript (\gamma )=\sqbrac{\sum _{j=2}^m(i_j-i_{j-1})^2}^{1/2}
\end{equation*}
For $a,b\in x$ with $a<b$, a \textit{geodesic} from $a$ to $b$ is a path from $a$ to $b$ that has the shortest length \cite{gud13}. Clearly, if $a<b$, then there is at least one geodesic from $a$ to $b$. If $a,b\in x$ are comparable and $a<b$ say, then the \textit{distance} $d(a,b)$ is the length of a geodesic from $a$ to $b$ \cite{gud13}. It is shown in \cite{gud13} that if $a<c<b$, then $d(a,b)\le d(a,c)+d(c,b)$. Thus, the triangle inequality holds when applicable so $d(a,b)$ has the most important property of a metric. It is also shown in \cite{gud13} that a subpath of a geodesic is a geodesic.

For conciseness, let us refer to a vertex by its label. If there are $j$ geodesics from vertex $1$ to vertex $n$, we define the \textit{curvature} $K(n)$ at $n$ to be $K(n)=j-1$ \cite{gud13}. One might argue that the curvature should be a local property and should not depend so heavily on vertex $1$ which could be a considerable distance away. However, if there are a lot of geodesics from $1$ to $n$, then by the last sentence of the previous paragraph, there are also many geodesics from other vertices to $n$. Thus, the definition of curvature is not as dependent on the initial vertex $1$ as it first appears. Assuming that particles tend to move along geodesics, we see that $K(n)$ gives a measure of the tendency for vertex $n$ to be occupied. In this way, the geometry of $x\in\pscript '$ determines the gravity at each vertex of $x$.

\section{Inflationary Period} 
In this section and the next we shall refer to a vertex in a $c$-causet by its label. Thus, for $x\in\pscript '_n$ we have $x=\brac{1,2,\ldots ,n}$. According to recent observations and cosmological studies, immediately after the big bang the universe expanded exponentially during an inflationary period. We can describe this period by the $c$-causet $x_{0,0}$ with shell sequence $(1,2,4,8,\ldots ,2^n)$. We see that $x_{0,0}$ has $\sum _{i=0}^n2^i=2^{n+1}$ vertices. We can obtain an approximate value for $n$ by noting that the number of photons in the universe is approximately $10^{89}$. We assume that most of these photons were created during the inflationary period and that they still exist as the cosmic microwave background. We also assume that during the inflationary period, the only particles that existed were photons and possibly ``dark photons'' which result in dark energy. Moreover, including empty cells through which photons can move we estimate that the number of vertices produced during the inflationary period to be about $10^{93}$. Hence, $2^{n+1}\approx 10^{93}$ so that
\begin{equation*}
n=93\,\frac{ln 10}{ln 2}-1\approx 308
\end{equation*}
Postulating that each shell is filled in one Planck instant given by about $10^{-43}$ seconds, we conclude that the inflationary period lasted only about $3\times 10^{-41}$ seconds.

It has also been experimentally verified that the universe, in the large, is quite flat and homogeneous. This suggests that during the inflationary period, the curvature should be essentially zero. We now indicate why this is true for our present model by considering some small cases. As an example of a toy inflationary universe, suppose $x$ has shell sequence $(1,2,4,8,16)$. With shells delineated by semi-colons, we label the vertices as follows:
\begin{equation*}
(1;2,3;4,5,6,7;8,9,10,11,12,13,14,15;16,17,\ldots ,31)
\end{equation*}
Then $(1,2)$ and $(1,3)$ are geodesics, $d(1,2)=1$, $d(1,3)=2$ and $K(2)=K(3)=0$. Similarly, $(1,2,4)$, $(1,3,4)$, $(1,3,5)$, $(1,3,6)$ and $(1,3,7)$ are geodesics, $d(1,4)=\sqrt{5}$, $d(1,5)=\sqrt{8}$, $d(1,6)=\sqrt{13}$, $d(1,7)=\sqrt{20}$ and $K(4)=1$, $k(5)=K(6)=K(7)=0$. Table~1 displays the distances squared and curvatures for the vertices of $x$.
\vskip 2pc

{\hskip - 3pc
\begin{tabular}{c|c|c|c|c|c|c|c|c|c|c|c|c|c|c|c|c|c}
Vertex $i$&$2$&$3$&$4$&$5$&$6$&$7$&$8$&$9$&$10$&$11$&$12$&$13$&$14$&$15$&$16$&$17$&$18$\\
\hline
$d(1,i)^2$&$1$&$4$&$5$&$8$&$13$&$20$&$17$&$22$&$29$&$39$&$45$&$56$&$69$&$84$&$61$&$70$&$81$\\
\hline
$K(i)$&$0$&$0$&$1$&$0$&$0$&$0$&$1$&$0$&$1$&$0$&$0$&$0$&$0$&$0$&$1$&$0$&$1$\\
\end{tabular}}
\vskip 2pc
{\hskip - 3pc
\begin{tabular}{c|c|c|c|c|c|c|c|c|c|c|c|c|c}
Vertex $i$&$19$&$20$&$21$&$22$&$23$&$24$&$25$&$26$&$27$&$28$&$29$&$30$&$31$\\
\hline
$d(1,i)^2$&$92$&$105$&$118$&$133$&$148$&$165$&$184$&$205$&$228$&$253$&$280$&$309$&$340$\\
\hline
$K(i)$&$0$&$1$&$0$&$1$&$0$&$0$&$0$&$0$&$0$&$0$&$0$&$0$&$0$\\
\noalign{\bigskip}
\multicolumn{14}{c}%
{\textbf{Table 1 (Distances Squared and Curvatures)}}\\
\end{tabular}}
\vskip 2pc

\noindent Continuing this example to the next shell with $32$ vertices, we find that $K(i)=0$ except for
$i=32,34,36,38,40,42,44,46$ in which case $K(i)=1$. We conclude that in this case, the curvature is zero except for a sprinkling of ones which are about one-fourth the total.

That this conclusion always holds is reinforced by considering the geodesics for the $c$-causet with shell sequence
$(1,2,4,8,16,32)$. We denote a geodesic with vertices having labels $i_1,i_2,\ldots ,i_n$ by $(i_1,i_2,\ldots ,i_n)$. The geodesics terminating at vertices in shell~1 are: (1,2), (1,3). Those terminating at vertices in shell~2 are: (1,2,4), (1,3,4), (1,3,5), (1,3,6), (1,3,7).

\noindent For shell~3 we have: (1,3,5,8), (1,3,6,8), (1,3,6,9), (1,3,6,10), (1,3,7,10), (1,3,7,11), (1,3,7,12), (1,3,7,13), (1,3,7,14), (1,3,7,15).

\noindent For shell~4 we have: (1,3,7,12,16), (1,3,7,13,16), (1,3,7,12,17), (1,3,7,12,18), (1,3,7,13,18),
(1,3,5,13,19), (1,3,7,13,20), (1,3,7,14,20), (1,3,7,14,21), (1,3,7,14,22), (1,3,7,15,22), (1,3,7,15,23), (1,3,7,15,24),\ldots , (1,3,7,15,31).

\noindent For shell 5 we have: (1--15,23,32), (1--15,24,32), (1--15,24,33), (1--15,24,34), (1--15,25,34), (1--15,25,35), (1--15,25,36), (1--15,26,36), (1--15,26,37), (1--15,26,38), (1--15,27,38), (1--15,27,39), (1--15,27,40), (1--15,28,40), (1--15,28,41), (1--15,28,42),
(1--15,29,42), (1--15,29,43), (1--15,29,44), (1--15,30,44), (1--15,31,48), (1--15,31,49),\newline\ldots , (1--15,31,63).
\vskip 1pc

\noindent This pattern continues indefinitely and we conclude that the curvatures are all zero or one with approximately $3/4$ zeros and $1/4$ ones.

Let $m,m+1,\ldots ,m+n$ be the labels of the vertices in a shell for a $c$-causet. We call $m+n$ the \textit{endpoint} of the shell and $\floorbrac{m+\frac{n}{2}}$ the \textit{midpoint} where $\floorbrac{r}$ denotes the integer part of $r$. Let $x$ be an exponentially growing $c$-causet with shell sequence $(1,2,,,,4,\ldots ,2^n)$. It follows by induction on the shell number that every geodesic $(i_1,i_2,\ldots i_n)$,
$n\ge 3$, of $x$ has $i_1,i_2,\ldots ,i_{n-2}$ as endpoints. Also, $i_{n-1}$ is an endpoint if $i_n$ is at least as large as the midpoint of its shell and moreover, $i_{n-1}$ is at least as large as the midpoint of its shell. This gives a fairly complete description of the geodesics in $x$.

We have seen that exponentially growing $c$-causets have a uniformly low curvature. It is of interest that this is a characteristic property of exponential growth. That is, a $c$-causet that grows at a slower pace does not have this property. To illustrate this, let
$y$ be the $c$-causet with shell sequence $(1,2,3,\ldots ,10)$ so $y$ is growing very slowly. Notice that $y$ has $55$ vertices so $y$ is about the same size as the exponentially growing $c$-causet considered previously. However, the curvatures are considerably different as the following table shows.
\vskip 2pc

{\hskip -2pc
\begin{tabular}{c|c|c|c|c|c|c|c|c|c|c|c|c|c|c|c|c}
Vertex $i$&2&3&4&5&6&7&8&9&10&11&12&13&14&15&16&17\\
\hline
$K(i)$&0&0&1&0&0&0&1&0&0&2&0&1&0&0&2&2\\
\end{tabular}}
\vskip 2pc
{\hskip - 1.5pc
\begin{tabular}{c|c|c|c|c|c|c|c|c|c|c|c|c|c|c}
Vertex $i$&18&19&20&21&22&23&24&25&26&27&28&29&30&31\\
\hline
$K(i)$&0&1&0&0&3&2&1&0&0&0&0&2&0&1\\
\end{tabular}}
\vskip 2pc
{\hskip - 1.5pc
\begin{tabular}{c|c|c|c|c|c|c|c|c|c|c|c|c|c}
Vertex $i$&32&33&34&35&36&37&38&39&40&41&42&43&44\\
\hline
$K(i)$&1&0&1&0&0&3&2&0&1&2&0&0&0\\
\end{tabular}}
\vskip 2pc
{\hskip - 1.5pc
\begin{tabular}{c|c|c|c|c|c|c|c|c|c|c|c}
Vertex $i$&45&46&47&48&49&50&51&52&53&54&55\\
\hline
$K(i)$&0&2&4&3&0&0&1&0&0&0&0\\
\noalign{\bigskip}
\multicolumn{12}{c}%
{\textbf{Table 2 (Curvatures in $c$-causet $y$)}}\\
\end{tabular}}
\vskip 2pc

The inflationary period lasts until the system arrives at the $c$-causet $x_{0,0}$ with shell sequence $(1,2,4,\ldots ,2^n)$ with
$n\approx 308$. For this $n$ the system can no longer sustain the immense energies, pressures and densities. The inflationary period terminates and the system goes through a phase transition. At this point most of the vertices (cells) of a universe have been produced and the exponential growth ceases. The system resembles a huge sea of vertices and the vertex production continues above this sea at a much slower pace. The model then enters what we call the \textit{multiverse period}.  This is the period in which we humans now find ourselves.

\section{Multiverse Period} 
The multiverse period begins with the $c$-causet $x_{0,0}$. Then the system diverges in various ways forming different $c$-causet paths starting at the common $c$-causet $x_{0,0}$. Each path corresponds to a distinct universe history and these evolving universes form a multiverse model. During the inflationary period, there are no physical principles operating except the geometry given by the shell sequence. In particular, the curvature is essentially zero and there is no concept of dimension. As we shall see, at the beginning of the multiverse period, the structure of a 4-dimensional discrete manifold emerges and quantum mechanics enters the picture. In this sense, the much sought after ``theory of everything'' is quantum mechanics itself. Moreover, very early in the multiverse period, high curvatures are created which explains why matter collects in certain places and not in others. However, these curvatures are local and since most of the vertices of any of the universes have already been formed during the inflationary period, the universes remain relatively flat. This is because once a curvature has been established it remains unchanged at later times. In the rest of this section and the next, we shall make these statements precise.

An element $a$ of a causet $x$ is \textit{maximal} if there is no $b\in x$ with $a<b$. Thus, $a$ is maximal in $x$ if there are no elements of $x$ that are in the causal future of $a$. If $x,y\in\pscript '$ we say that $x$ \textit{produces} $y$ if $y$ is obtained by first adjoining a maximal element $a$ to $x$ and then adjoining a second maximal element $b$ to $x\cup\brac{a}$ so that
$y=x\cup\brac{a,b}$ where $a\notin x$, $b\notin x$. If $x$ produces $y$, we write $x\to y$ and say that $y$ is an \textit{offspring} of $x$ It is easy to see that any $x\in\pscript '$ has precisely four distinct offspring. In fact, if $x$ has shell sequence
$(s_o,s_1,\ldots ,s_n)$, then the offspring of $x$ have shell sequences:
\begin{align}         
\label{eq31}
&(s_0,s_1,\ldots ,s_n+2)\notag\\
&(s_0,s_1,\ldots ,s_n+1,1)\notag\\
&(s_0,s_1,\ldots ,s_n,2)\\
&(s_0,s_1,\ldots ,s_n,1,1)\notag
\end{align}
We number the $c$-causets in the multiverse period recursively as follows. The four offspring of $x_{0,0}$ in the order given by \eqref{eq31} are denoted by $x_{1,0},x_{1,1},x_{1,2},x_{1,3}$. Given $c$-causet $x_{n,j}$ the offspring in the order given by \eqref{eq31} are $x_{n+1,4j+k}$, $k=0,1,2,3$. For example, the offspring of $x_{1,0},x_{1,1},x_{1,2},x_{1,3}$ are:
\begin{align*}
\label{eq31}
&x_{2,0},x_{2,1},x_{2,2},x_{2,3}\\
&x_{2,4},x_{2,5},x_{2,6},x_{2,7}\\
&x_{2,8},x_{2,9},x_{2,10},x_{2,11}\\
&x_{2,12},x_{2,13},x_{2,14},x_{2,15}
\end{align*}
respectively.

We use the notation $\pscript _0=\brac{x_{0,0}}, \pscript _1=\brac{x_{1,0},x_{1,1},x_{1,2},x_{1,3}},\ldots ,$
\begin{equation*}
\pscript _n=\brac{x_{n,j}\colon j=0,1,\ldots ,4^n-1}
\end{equation*}
We thus see that $\ab{\pscript _n}=4^n$. We denote the set of all $c$-causets in the multiverse period by
$\pscript =\cup _{n=0}^\infty\pscript _n$. We can view $\pscript$ as a 4-dimensional discrete manifold in which the ``tangent vectors'' at $x_{n,j}$ consist of the four edges
\begin{equation*}
d_{n,j}^k=(x_{n,j},x_{n+1,4j+k}),\quad k=0,1,2,3
\end{equation*}
We think of $d_{n,j}$, $k=0,1,2,3$, as the four directions originating at $x_{n,j}$. For simplicity, we consider $x_{0,0}$ as the $c$-causet with a single vertex that we denote by $0$. Thus, we identify $x_{0,0}$ with the $c$-causet $\brac{0}$. In this way w have
$\ab{x_{0,0}}=1$ and the shell sequence of $x_{0,0}$ is simply $(1)$. In a similar way $\ab{x_{1,j}}=3$, $j=0,1,2,3$ and the shell sequence of $x_{1,j}$, $j=0,1,2,3$ are $(3),(2,1),(1,2),(1,1,1)$, respectively. In general $\ab{x_{n,j}}=2n+1$ and the shell sequences of $x_{n,j}$, $=0,1,\ldots ,4^n-1$ are
\begin{equation*}
(2n+1),(2n,1),(2n-1,2),(2n-1,1,1),\ldots ,(1,1,\ldots 1)
\end{equation*}
respectively.

For the multiverse period, we define paths, distances, geodesics and curvature as before except now we begin with the vertex $0$. To show how large curvatures can be generated, we consider an example of an ideal pulsating universe $x_p$. Although our universe may not be described exactly by $x_p$, we are proposing that it may be similar to $x_p$. The $c$-causet $x_p$ has shell sequence
\begin{equation}         
\label{eq32}
(1,2,4,2,4,6,4,2,4,6,8,6,4,2,\ldots ,6,4,2)
\end{equation}
As before, we denote the vertices by their labels. The vertices are listed as follows where we stop at vertex 72:
\begin{equation*}
(0;1,2;3,4,5,6;7,8;9,10,11,12;\cdots ;65,66,67,68,69,70,71,72)
\end{equation*}
Examples of two geodesics are $(0,2,4,7)$ and $(0,2,5,7)$. It follows that $d(0,7)=\sqrt{17}$ and $K(7)=1$. Also, $(0,2,5,8)$ is a geodesic, $d(0,8)=\sqrt{22}$ and $k(8)=0$. Table~3 displays the distances squared and curvatures for the vertices of $x_p$ up to 72.
\vskip 2pc

{\hskip - 1.5pc
\begin{tabular}{c|c|c|c|c|c|c|c|c|c|c|c|c|c|c|c}
Vertex $i$&1&2&3&4&5&6&7&8&9&10&11&12&13&14&15\\
\hline
$d(0,i)^2$&1&4&5&8&13&20&17&22&21&26&31&38&35&40&47\\
\hline
$K(i)$&0&0&1&0&0&0&1&0&1&$2$&0&0&2&0&1\\
\end{tabular}}
\vskip 2pc
{\hskip - 1.5pc
\begin{tabular}{c|c|c|c|c|c|c|c|c|c|c|c|c|c|c}
Vertex $i$&16&17&18&19&20&21&22&23&24&25&26&27&28&29\\
\hline
$d(0,i)^2$&54&63&74&63&70&79&88&79&86&83&88&95&102&97\\
\hline
$K(i)$&0&0&0&2&0&1&0&3&0&3&3&4&0&3\\
\end{tabular}}
\vskip 2pc
{\hskip - 1.5pc
\begin{tabular}{c|c|c|c|c|c|c|c|c|c|c|c}
Vertex $i$&30&31&32&33&34&35&36&37&38&39&40\\
\hline
$d(0,i)^2$&104&111&118&127&138&127&134&143&152&163&174\\
\hline
$K(i)$&8&5&0&0&0&6&0&1&0&1&0\\
\end{tabular}}
\vskip 2pc
{\hskip - 1.5pc
\begin{tabular}{c|c|c|c|c|c|c|c|c|c|c|c}
Vertex $i$&41&42&43&44&45&46&47&48&49&50&51\\
\hline
$d(0,i)^2$&187&202&177&188&199&210&223&236&213&224&235\\
\hline
$K(i)$&0&0&0&1&1&0&1&0&2&3&2\\
\end{tabular}}
\vskip 2pc
{\hskip - 1.5pc
\begin{tabular}{c|c|c|c|c|c|c|c|c|c|c|c}
Vertex $i$&52&53&54&55&56&57&58&59&60&61&62\\
\hline
$d(0,i)^2$&246&229&238&233&238&245&254&247&254&261&270\\
\hline
$K(i)$&0&2&2&2&2&2&5&2&5&2&8\\
\end{tabular}}
\vskip 2pc
{\hskip - 1.5pc
\begin{tabular}{c|c|c|c|c|c|c|c|c|c|c}
Vertex $i$&63&64&65&66&67&68&69&70&71&72\\
\hline
$d(0,i)^2$&279&290&277&286&295&304&315&326&339&354\\
\hline
$K(i)$&5&11&2&11&14&5&17&11&11&11\\
\noalign{\bigskip}
\multicolumn{10}{c}%
{\textbf{Table 3 (Distances Squared and Curvatures for $x_p$)}}\\
\end{tabular}}
\vskip 2pc

We call the set of vertices in $x_p$ between 2s in the shell sequence \eqref{eq32} a \textit{pulse}. Starting with the first 2 in \eqref{eq32}, let $p_0,p_1,p_2,\ldots$ be the pulses. That is, $p_0=\brac{1,2,\ldots ,7,8}$, $p_1=\brac{7,8,\ldots ,23,24}$,
$p_2=\brac{23,24,\ldots ,53,54},\ldots\,$. We see that $\ab{p_0}=8,\ab{p_1}=18,\ab{p_2}=32,\ab{p_3}=50,\ldots\,$. In general we have that
\begin{equation*}
\ab{p_n}=8+10n+4\sqbrac{1+2+\cdots +(n-1)}=8+10n+2(n-1)n=2(n+2)^2
\end{equation*}
As previously discussed, during the multiverse period, we are assuming that two vertices are created every Planck instant. It follows that the time duration of the $n$th pulse in Planck instances is
\begin{equation*}
t_n=\tfrac{1}{2}\ab{p_n}-1=(n+2)^2+1=(n+1)(n+3)
\end{equation*}
Thus, $t_o=3,t_1=8, t_2=15,t_3=24,\ldots\,$. The time $\tau _n$ in Planck instants at which the $n$th pulse begins is computed as follows:
\begin{equation*}
\tau _o=1, t_1=1+t_0=4,\tau _2=1+t_0+t_1=12,\tau_3=1+t_0+t_1+t_2=27,\ldots
\end{equation*}
In general, we have that
\begin{align*}
\tau _n&=1+\sum _{i=0}^{n-1}t_i=1+\sum _{i=0}^{n-1}(i+1)(i+3)=1+\sum _{i=0}^{n-1}(3+4i+i^2)\\
  &=1+3n+2(n-1)n+\frac{(n-1)n(2n-1)}{6}\\
  &=1+n(2n+1)+\frac{(n-1)n(2n-1)}{6}
\end{align*}

The time of our present is approximately $10^{60}$ Planck instants. If we are in the $n$th pulse now, then $\tau _n\approx 10^{60}$. Since $n$ is large, we have that $\tau _n\approx n^3/3$. Hence, $n^3/3\approx 10^{60}$ so that $n\approx 3^{1/3}\times 10^{20}$. We conclude that the universe $x_p$ is about in its $10^{20}$th pulse. The duration of this pulse is $t_n\approx 10^{40}$ Planck instants which is approximately $10^{-3}$ seconds. If this is near the behavior of our own universe, it would be close to being imperceptible. Finally, the number of new vertices produced in each universe during the multiverse period is about $2\times 10^{60}$ This is much smaller than the number of vertices produced during the inflationary period. It has been estimated that the number of massive particles in our universe is about $10^{80}$. Some of these particles were produced during the multiverse period, but most of them came from photons produced during the inflationary period.

\section{Quantum Mechanics Emerges} 
This section discusses the emergence of quantum mechanics at the beginning of the multiverse period. The set $\pscript$ of odd cardinality $c$-causets together with the production relation $\shortrightarrow$ gives a tree $(\pscript ,\shortrightarrow )$. We have seen in Section~3 that we can denote the $c$-causets in $\pscript$ by $x_{n,j}$, $n=0,1,2,\ldots$, $j=0,1,\ldots ,4^n-1$. Moreover, we have four directions (tangent vectors)
\begin{equation*}
d_{n,j}^k=(x_{n,j},x_{n+1,4j+k}),\quad k=0,1,2,3
\end{equation*}
emenating from each node ($c$-causet) $x_{n,j}$. We say that two tangent vectors are \textit{incident} if they have the form $(x,y)$ and $(y,z)$. We also use the notation
\begin{equation*}
\pscript _n=\brac{x_{n,0},x_{n,1},\ldots ,x_{n,4^n-1}}
\end{equation*}
Each $x\in\pscript$ except $x_{0,0}$ has a unique producer and each $x\in\pscript$ has exactly four offspring. In particular
$x_{n,j}\to x_{n+1,4j+k}$, $k=0,1,2,3$ and we interpret the tree $(\pscript ,\shortrightarrow )$ as a sequential growth process.

An $n$-\textit{path} in $\pscript$ is a sequence $\omega =\omega _0\omega _1\cdots\omega _n$ where $\omega _i\in\pscript _i$ and
$\omega _i\to\omega _{i+1}$. We denote the set of $n$-paths by $\Omega _n$ and interpret an $\omega\in\Omega$ as the history of the universe $\omega _n$ at which $\omega$ terminates. We can also consider an $n$-path as a sequence of tangent vectors
\begin{equation*}
\omega =d_{0,0}^{K_0}d_{1,j_1}^{k_1}\cdots d_{n-1,j_{n-1}}^{k_{n-1}}
\end{equation*}
where each tangent vector is incident to the next. Since each $x_{n,j}\in\pscript _n$ has a unique history, we can identify $\pscript _n$ with $\Omega _n$ and we write $\pscript _n\approx\Omega _n$.

We now describe the evolution of the multiverse in terms of a quantum sequential growth process. In such a process, the probabilities of competing geometries are determined by quantum amplitudes. These amplitudes provide interferences that are characteristic of quantum systems. A \textit{transition amplitude} is a map $\atilde\colon\pscript\times\pscript\to\complex$ satisfying $\atilde (x,y)=0$ if
$x\not\to y$ and$\sum _{y\in\pscript}\atilde (x,y)=1$ for all $x\in\pscript$. We call $\atilde$ a \textit{unitary transition amplitude} (uta) if
$\atilde$ also satisfies $\sum _{y\in\pscript}\ab{\atilde (x,y)}^2=1$ for all $x\in\pscript$. We conclude that a uta satisfies
\begin{equation}         
\label{eq41}
\sum _{k=0}^3\atilde (x_{n,j},x_{n+1,4j+k})=\sum _{k=0}^3\ab{\atilde (x_{n,j},x_{n+1,4j+k})}^2=1
\end{equation}
for all $n=0,1,2,\ldots$, $j=0,1,\ldots ,4^n-1$. Using the obvious notation, we can also write \eqref{eq41} as
\begin{equation*}
\sum _{k=0}^3\atilde (d_{n,j}^k)=\sum _{k=0}^3\ab{\atilde (d_{n,j}^k)}^2=1
\end{equation*}

It is shown in \cite{gud142,gud15} how utas can be constructed. We call
\begin{equation*}
c_{n,j}^k=\atilde (x_{n,j},x_{n+1,4j+k})=\atilde (d_{n,j}^k),\quad k=0,1,2,3
\end{equation*}
the \textit{coupling constants} for $\atilde$. As of now, the coupling constants are undetermined. It is hoped that their values can be found so that general relativity becomes an approximation to the present theory. We have thus obtained a kind of quantum Markov chain in which transition probabilities are replaced by utas. If $\atilde$ is a uta and
$\omega =\omega _0\omega _1\cdots\omega _n\in\Omega _n$, we define the \textit{amplitude} of $\omega$ to be
\begin{equation*}
a(\omega )=\atilde (\omega _0,\omega _1)\atilde (\omega _1,\omega _2)\cdots\atilde (\omega _{n-1},\omega _n)
\end{equation*}
Also, we define the \textit{amplitude} of $x\in\pscript _n$ to be $a(\omega )$ where $\omega$ is the unique path in $\Omega _n$ that terminates at $x$.

Let $H_n$ be the Hilbert space
\begin{equation*}
H_n=L_2(\Omega _n)=L_2(\pscript _n)=\brac{f\colon\pscript _n\to\complex}
\end{equation*}
with the standard inner product
\begin{equation*}
\elbows{f,g}=\sum _{x\in\pscript _n}\overline{f(x)}g(x)
\end{equation*}
Let $\xhat _{n,j}$ be the unit vector in $H_n$ given by the characteristic function $\chi _{x_{n,j}}$. Thus, $\dim H_n=4^n$ and
$\brac{\xhat _{n,j}\colon j=0,1,\ldots ,4^n-1}$ becomes an orthonormal basis for $H_n$. For the remainder of this section, $\atilde$ is a uta with corresponding coupling constants $c_{n,j}^k$. We now describe the quantum dynamics generated by $\atilde$. Define the operators $U_n\colon H_n\to H_{n+1}$ by
\begin{equation*}
U_n\xhat _{n,j}=\sum _{k=0}^3c_{n,j}^k\xhat _{n+1,4j+k}
\end{equation*}
and extend $U_n$ to $H_n$ by linearity. It is shown in \cite{gud142,gud15} that $U_n$ is a partial isometry with $U_n^*U_n=I_n$ where the adjoint $U_n^*\colon H_{n+1}\to H_n$ of $U_n$ is given by
\begin{equation*}
U_n^*\xhat _{n+1,4j+k}=\cbar _{n,j}^{\,k}\xhat _{n,j},\quad k=0,1,2,3
\end{equation*}

As usual, a \textit{state} on $H_n$ is a positive operator $\rho$ on $H_n$ with $\rmtr (\rho )=1$. A \textit{stochastic state} on $H_n$ is a state $\rho$ that satisfies $\elbows{\rho 1_n,1_n}=1$ where $1_n=\chi _{\pscript _n}$; that is, $1_n(x)=1$ for all $x\in\pscript _n$. It is shown in \cite{gud15} that if $\rho$ is a state on $H_n$, then $U_n\rho U_n^*$ is a state on $H_{n+1}$. Moreover, if $\rho$ is a stochastic state on $H_n$, then $U_n\rho U_n^*$ is a stochastic state on $H_{n+1}$. The quantum dynamics is then given by the map
$\rho\mapsto U_n\rho U_n^*$. For further details, we refer the reader to \cite{gud15}.

Although $U_n\colon H_n\to H_{n+1}$ is a partial isometry, it cannot be unitary because $H_n$ and $H_{n+1}$ have different dimensions. However, when an additional condition is imposed, we can construct a related unitary operator. We say that $\atilde$ and $c_{n,j}^k$ are \textit{strong} if the matrices $K_{n,j}$ given by
\begin{equation}         
\label{eq42}
K_{n,j}=\begin{bmatrix}\noalign{\smallskip}c_{n,j}^0&c_{n,j}^1&c_{n,j}^2&c_{n,j}^3\\\noalign{\smallskip}
c_{n,j}^1&c_{n,j}^0&c_{n,j}^3&c_{n,j}^2\\\noalign{\smallskip}c_{n,j}^2&c_{n,j}^3&c_{n,j}^0&c_{n,j}^1\\\noalign{\smallskip}
c_{n,j}^3&c_{n,j}^2&c_{n,j}^1&c_{n,j}^0\\\noalign{\smallskip}\end{bmatrix}
\end{equation}
are unitary. Notice that if $\atilde$ is a uta, then $\sum _{k=0}^3\ab{c_{n,j}^k}^2=1$ so $K_n,j$ is unitary if and only if
\begin{align}         
\label{eq43}
\rmre (c_{n,j}^0\cbar _{n,j}^{\,1}+c_{n,j}^2\cbar _{n,j}^{\,3})&=\rmre (c_{n,j}^0\cbar _{n_j}^{\,2}+c_{n,j}^1\cbar _{n,j}^{\,3})\notag\\
  &=\rmre (c_{n,j}^0\cbar _{n_j}^{\,3}+c_{n,j}^1\cbar _{n,j}^{\,2})=0
\end{align}
Thus, \eqref{eq43} is a necessary and sufficient condition for a uta to be strong. A uta already satisfies a condition similar to \eqref{eq43} but not quite as strong. Since $\sum _{k=0}^3c_{n,j}^k=1$ we have
\begin{equation*}
1=\ab{\sum _{k=0}^3c_{n,j}^k}^2=\sum _{k=0}^3c_{n,j}^k\sum _{k'=0}^3\cbar _{n,j}^{\,k'}
  =\sum _{k=0}^3\ab{c_{n,j}^k}^2+2\rmre\sum _{{k,k'=0}\atop{k<k'}}^3c_{n,j}^k\cbar _{n,j}^{\,k'}
\end{equation*}
It follows that
\begin{equation}         
\label{eq44}
\rmre\sum _{{k,k'=0}\atop{k<k'}}^3c_{n,j}^k\cbar _{n,j}^{\,k'}=0
\end{equation}
Notice that \eqref{eq43} implies \eqref{eq44}. It follows that if $K_{n,j}$ is unitary, then $\atilde$ is essentially a uta. In fact when
\begin{equation*}
\sum _{k=0}^3\ab{c_{n,j}^i}^2=\ab{\sum _{k=0}^3c_{n,j}^k}=1
\end{equation*}
If we multiply the $c_{n,j}^k$ by a fixed phase factor $e^{i\phi _{n,j}}$, then $e^{i\phi _{n,j}}c_{n,j}$ becomes a uta.

A direct verification shows that the eigenvalues of $K_{n,j}$ are
\begin{align}         
\label{eq45}
\lambda _{n+1,j}^0&=c_{n,j}^0+c_{n,j}^1+c_{n,j}^2+c_{n,j}^3=1\notag\\
\lambda _{n+1,j}^1&=c_{n,j}^0-c_{n,j}^1+c_{n,j}^2-c_{n,j}^3\notag\\
\lambda _{n+1,j}^2&=c_{n,j}^0+c_{n,j}^1-c_{n,j}^2-c_{n,j}^3\\
\lambda _{n+1,j}^3&=c_{n,j}^0-c_{n,j}^1-c_{n,j}^2+c_{n,j}^3\notag
\end{align}
with corresponding unit eigenvectors
\begin{equation*}
\tfrac{1}{2}(1,1,1,1),\tfrac{1}{2}(1,-1,1,-1),\tfrac{1}{2}(1,1,-1,-1),\tfrac{1}{2}(1,-1,-1,1),
\end{equation*}
respectively.

The Hilbert space $H_{n+1}$ can be decomposed into the direct sum
\begin{equation}         
\label{eq46}
H_{n+1}=H_{n+1,0}\oplus H_{n+1,1}\oplus\cdots\oplus H_{n+1,4^n-1}
\end{equation}
where $\dim H_{n+1,j}=4$ and an orthonormal basis for $H_{n+1,j}$, $n=0,1,2\ldots$, is
\begin{equation*}
\brac{\xhat _{n+1,4j+k}\colon k=0,1,2,3}
\end{equation*}
We define a stochastic unitary operator $V_{n+1,j}\colon H_{n+1,j}\to H_{n+1,j}$ with matrix representation $K_{n,j}$ given by
\eqref{eq42}. To be explicit, we have
\begin{align*}
V_{n+1,j}\xhat _{n+1,4j}
&=c_{n,j}^0\xhat _{n+1,4j}+c_{n,j}^1\xhat _{n+1,4j+1}+c_{n,j}^2\xhat _{n+1,4j+2}+c_{n,j}^3\xhat _{n+1,4j+3}\\
V_{n+1,j}\xhat _{n+1,4j+1}
&=c_{n,j}^1\xhat _{n+1,4j}+c_{n,j}^0\xhat _{n+1,4j+1}+c_{n,j}^3\xhat _{n+1,4j+2}+c_{n,j}^2\xhat _{n+1,4j+3}\\
V_{n+1,j}\xhat _{n+1,4j+2}
&=c_{n,j}^2\xhat _{n+1,4j}+c_{n,j}^3\xhat _{n+1,4j+1}+c_{n,j}^0\xhat _{n+1,4j+2}+c_{n,j}^1\xhat _{n+1,4j+3}\\
V_{n+1,j}\xhat _{n+1,4j+3}
&=c_{n,j}^3\xhat _{n+1,4j}+c_{n,j}^2\xhat _{n+1,4j+1}+c_{n,j}^1\xhat _{n+1,4j+2}+c_{n,j}^0\xhat _{n+1,4j+3}\\
\end{align*}
We conclude from our previous work that the eigenvalues of $V_{n+1,j}$ are $\lambda_{n+1,j}^k$, $k=0,1,2,3$ given by
\eqref{eq45} with corresponding unit eigenvectors
\begin{align}         
\label{eq47}
\xhat _{n+1,j}^0&=\tfrac{1}{2}(\xhat _{n+1,4j}+\xhat _{n+1,4j+1}+\xhat _{n+1,4j+2}+\xhat _{n+1,4j+3})\notag\\
\xhat _{n+1,j}^1&=\tfrac{1}{2}(\xhat _{n+1,4j}-\xhat _{n+1,4j+1}+\xhat _{n+1,4j+2}-\xhat _{n+1,4j+3})\notag\\
\xhat _{n+1,j}^2&=\tfrac{1}{2}(\xhat _{n+1,4j}+\xhat _{n+1,4j+1}-\xhat _{n+1,4j+2}-\xhat _{n+1,4j+3})\\
\xhat _{n+1,j}^3&=\tfrac{1}{2}(\xhat _{n+1,4j}-\xhat _{n+1,4j+1}-\xhat _{n+1,4j+2}+\xhat _{n+1,4j+3})\notag
\end{align}
Finally, we define the stochastic unitary operator $V_{n+1}$ on $H_{n+1}$ by
\begin{equation*}
V_{n+1}=V_{n+1,0}\oplus V_{n+1,1}\oplus\cdots\oplus V_{n+1,4^n-1}
\end{equation*}
The eigenvalues of $V_{n+1}$ are $1$ (with multiplicity $4^n$), $\lambda _{n,j}^1,\lambda _{n,j}^2,\lambda _{n,j}^3$,
$j=0,1,\ldots ,4^n-1$. The corresponding unit eigenvectors are $\xhat _{n+1,j}^k$, $j=0,1,\ldots ,4^n-1$, $k=0,1,2,3$ given by
\eqref{eq47}. The operator $V_{n+1}$ provides an intrinsic symmetry on $H_{n+1}$ generated by the coupling constants. We call
$V_{n+1}$ the \textit{coupling constant symmetry} operator. In this framework, we can also define decoherence functionals and quantum measures in a standard way \cite{gud13,gud141,sor94}.

\section{Dirac Operators} 
In previous sections we have represented the $c$-causets in $\pscript _n$ by $x_{n,j}$, $n=0,1,\ldots$, $j=0,1,\ldots ,4^n-1$. We then denoted the standard orthonormal bases for the Hilbert spaces $H_n$ by $\xhat _{n,j}$. In this section, we follow Dirac to define an energy-momentum operator on $H_n$. To accomplish this, it is useful to write $j$ in its quartic representation
\begin{equation*}
j=j_{n-1}4^{n-1}+j_{n-2}4^{n-2}+\cdots +j_14+j_0,\quad j_i\in\brac{0,1,2,3}
\end{equation*}
which as usual, we abbreviate
\begin{equation}         
\label{eq51}
j=j_{n-1}j_{n-2}\cdots j_1j_0,\quad j_i\in\brac{0,1,2,3}
\end{equation}
This representation describes the directions that a path turns when moving from $x_{0,0}$ to $x_{n,j}$. For example, \eqref{eq51} represents the path that turns in direction $j_{n-1}$ at $x_{0,0}$, then turns in direction $j_{n-2},\ldots$, and finally turns in direction $j_0$ before arriving at $x_{n,j}$. If $j$ has the form \eqref{eq51} and $k\in\brac{1,2,3}$, define
\begin{equation*}
j^k=j'_{n-1}j'_{n-2}\cdots j'_1j'_0
\end{equation*}
where $j'_i=j_i$ if $j_i=k$ and $j'_i=0$ if $j_i\ne k$. We also define $j^0=j$. Notice that $j^0=j^1+j^2+j^3$.

On $H_n$ define the operators $Q_n^k\xhat _{n,j}=\sqrt{j^k}\xhat _{n,j}$ and extend by linearity. We then have that
\begin{equation*}
(Q_n^0)^2=(Q_n^1)^2+(Q_n^2)^2+(Q_n^3)^2
\end{equation*}
Thinking of the $Q_n^k$ as coordinate operators, we define the energy-momentum operators $P_n^k=V_nQ_n^kV_n^*$, $k=0,1,2,3$ on $H_n$. In $\real ^4$ define the indefinite inner product
\begin{equation*}
\qunderbar\ctimes\punderbar =q_0p_0-q_1p_1-q_2p_2-q_3p_3
\end{equation*}
and let $\sigma _k$, $k=0,1,2,3$, be the Pauli matrices
\begin{equation*}
\sigma _0=\begin{bmatrix}1&0\\0&1\end{bmatrix},\quad\sigma _1=\begin{bmatrix}0&1\\1&0\end{bmatrix},\quad
\sigma _2=\begin{bmatrix}0&-i\\i&0\end{bmatrix},\quad\sigma _3=\begin{bmatrix}1&0\\0&-1\end{bmatrix}
\end{equation*}
Also define $\sigmaunderbar =(\sigma _0,\sigma _1,\sigma _2,\sigma _3)$ and $\cappunderbar _n=(P_n^0,P_n^1,P_n^2,P_n^3)$. We then have
\begin{align*}
\sigmaunderbar\ctimes\cappunderbar _{\,n}&=\sigma _0P_n^0-\sigmaunderbar _2P_n^2-\sigmaunderbar _3P_n^3\\
&=\begin{bmatrix}\noalign{\smallskip}P_n^0-P_n^3&-P_n^1+iP_n^2\\\noalign{\smallskip}-P_n^1-iP_n^2&P_n^0+P_n^3\\
\noalign{\smallskip}\end{bmatrix}
\end{align*}
We consider $\sigmaunderbar\ctimes\cappunderbar _n$ as a self-adjoint operator on $H_n\otimes\complex ^2$ and call
$\sigmaunderbar\ctimes\cappunderbar _n$ the \textit{discrete Weyl operator}. To find the eigenvalues and eigenvectors of
$\sigmaunderbar\ctimes\cappunderbar _n$ notice that the eigenvectors of $P_n^k$ are $\yhat _{n,j}=V_n\xhat _{n,j}$ with corresponding eigenvalues $\sqrt{j^k}$.
 
\begin{thm}       
\label{thm51}
Define $\phi _{n,j}=(\phi _{n,j}^1,\phi _{n,j}^2)$ and $\psi _{n,j}=(\psi _{n,j}^1,\psi _{n,j}^2)$, $j=0,1,\ldots ,4^n-1$, in
$H_n\otimes\complex ^2$ by
\begin{align*}
\phi _{n,j}^1&=\paren{\sqrt{j^0}+\sqrt{j^3}}\yhat _{n,j},\quad \phi _{n,j}^2=\paren{\sqrt{j^1}+i\sqrt{j^2}}\yhat _{n,j}\\
\psi _{n,j}^1&=\paren{\sqrt{j^1}-i\sqrt{j^2}}\yhat _{n,j},\quad \psi _{n,j}^2=-\phi _{n,j}^1
\end{align*}
Then $\phi _{n,j}$ are eigenvectors of $\sigmaunderbar\ctimes\cappunderbar _n$ with eigenvalue $0$ and $\psi _{n,j}$ are eigenvectors of $\sigmaunderbar\ctimes\cappunderbar _n$ with eigenvalues $2\sqrt{j^0}$.
\end{thm}
\begin{proof}
A direct verification gives
\begin{align*}
&(P_n^0-P_n^3)\phi _{n,j}^1+(-P_n^1+iP_n^2)\phi _{n,j}^2\\
  &\quad =\paren{\sqrt{j^0}+\sqrt{j^3}}(P_n^0-P_n^3)\yhat _{n,j}+\paren{\sqrt{j^1}+i\sqrt{j^2}}(-P_n^1+iP_n^2)\yhat _{n,j}\\ 
  &\quad =\sqbrac{\paren{\sqrt{j^0}+\sqrt{j^3}}\paren{\sqrt{j^0}-\sqrt{j^3}}+\paren{\sqrt{j^1}+i\sqrt{j^2}}\paren{-\sqrt{j^1}+i\sqrt{j^2}}}\yhat _{n,j}\\
  &\quad =(j^0-j^1-j^2-j^3)\yhat _{n,j}=0
\intertext{and}
&-(P_n^1+iP_n^2)\phi _{n,j}^1+(P_n^0+P_n^3)\phi _{n,j}^2\\
  &\quad =\paren{\sqrt{j^0}+\sqrt{j^3}}(-P_n^1-iP_n^2)\yhat _{n,j}+\paren{\sqrt{j^1}+i\sqrt{j^2}}(-P_n^0+iP_n^3)\yhat _{n,j}\\ 
&\quad =\sqbrac{\paren{\sqrt{j^0}+\sqrt{j^3}}\paren{-\sqrt{j^1}-i\sqrt{j^2}}+\paren{\sqrt{j^1}+i\sqrt{j^2}}\paren{\sqrt{j^0}+i\sqrt{j^3}}}\yhat _{n,j}\\
  &\quad =0
\end{align*}
As before, we have
\begin{align*}
&(P_n^0-P_n^3)\psi _{n,j}^1+(-P_n^1+iP_n^2)\psi _{n,j}^2\\
  &\quad =\paren{\sqrt{j^1}-i\sqrt{j^2}}\paren{\sqrt{j^0}-\sqrt{j^3}}\yhat _{n,j}
     +\paren{-\sqrt{j^0}-\sqrt{j^3}}\paren{-\sqrt{j^1}+i\sqrt{j^2}}\yhat _{n,j}\\ 
  &\quad =2\paren{\sqrt{j^0j^1}-i\sqrt{j^0j^2}}\yhat _{n,j}=2\sqrt{j^0}\ \psi _{n,j}^1\\
\intertext{and}
&-(P_n^1+iP_n^2)\psi _{n,j}^1+(P_n^0+P_n^3)\psi _{n,j}^2\\
  &\quad =\paren{\sqrt{j^1}-i\sqrt{j^2}}\paren{-\sqrt{j^1}-i\sqrt{j^2}}\yhat _{n,j}
    +\paren{-\sqrt{j^0}-\sqrt{j^3}}\paren{\sqrt{j^0}+\sqrt{j^3}}\yhat _{n,j}\\ 
&\quad =\paren{-j^1-j^2-j^0-j^3-2\sqrt{j^0j^3}}\yhat _{n,j}\\
  &\quad =\paren{-2j^0-2\sqrt{j^0j^3}}\yhat _{n,j}=2\sqrt{j^0}\ \psi _{n,j}^2
\end{align*}
We conclude that $\sigmaunderbar\ctimes\cappunderbar _n\phi _{n,j}=0$ and 
$\sigmaunderbar\ctimes\cappunderbar _n\psi _{n,j}=2\sqrt{j^0}\ \psi _{n,j}$, $j=0,1,\ldots ,4^n-1$.
\end{proof}

The discrete Dirac operator is the 4-dimensional extension of the discrete Weyl operator. We first define the $4\times 4$ gamma matrices by
\begin{equation*}
\gamma _0=\begin{bmatrix}\sigma _0&0\\0&-\sigma _0\end{bmatrix},\quad
\gamma _k=i\begin{bmatrix}0&-\sigma _k\\\sigma _k&0\end{bmatrix},\quad k=1,2,3
\end{equation*}
Notice that $\gamma _j$ are self-adjoint matrices, $j=0,1,2,3$. The \textit{discrete Dirac operator} is
\begin{align*}
\gammaunderbar\ctimes\cappunderbar _{\,n}&=\gamma _0P_n^0-\gamma _1P_n^1-\gamma _2P_n^2-\gamma _3P_n^3\\
&=\begin{bmatrix}\noalign{\smallskip}P_n^0&0&iP_n^3&iP_n^1+P_n^2\\\noalign{\smallskip}
0&P_n^0&iP_n^1-P_n^2&-iP_n^3\\\noalign{\smallskip}-iP_n^3&-iP_n^1-P_n^2&-P_n^0&0\\\noalign{\smallskip}
-iP_n^1+P_n^2&iP_n^3&0&-P_n^0\\\noalign{\smallskip}\end{bmatrix}
\end{align*}
We interpret the self-adjoint operator $\gammaunderbar\ctimes\cappunderbar _n$ as the empty-momentum operator. To present the eigenvalues and eigenvectors of $\gammaunderbar\ctimes\cappunderbar _n$ we define
\begin{align*}
u_1&=\begin{bmatrix}\noalign{\smallskip}\paren{\sqrt{2}-1}\sqrt{j^0}\\\noalign{\smallskip}0\\i\sqrt{j^3}
\\\noalign{\smallskip}i\sqrt{j^1}-\sqrt{j^2}\\\noalign{\smallskip}\end{bmatrix},\quad
u_2=\begin{bmatrix}\noalign{\smallskip}0\\\noalign{\smallskip}\paren{\sqrt{2}-1}\sqrt{j^0}\\\noalign{\smallskip}
i\sqrt{j^1}+\sqrt{j^2}\\\noalign{\smallskip}-i\sqrt{j^3}\\\noalign{\smallskip}\end{bmatrix},\\\noalign{\smallskip}
u_3&=\begin{bmatrix}\noalign{\smallskip}i\sqrt{j^1}+\sqrt{j^2}\\\noalign{\smallskip}-i\sqrt{j^3}\\\noalign{\smallskip}
0\\\noalign{\smallskip}\paren{\sqrt{2}-1}\sqrt{j^0}\\\noalign{\smallskip}\end{bmatrix},\quad
u_4=\begin{bmatrix}\noalign{\smallskip}-i\sqrt{j^3}\\\noalign{\smallskip}i\sqrt{j^1}+\sqrt{j^2}\\\noalign{\smallskip}
\paren{\sqrt{2}-1}\sqrt{j^0}\\\noalign{\smallskip}0\\\noalign{\smallskip}\end{bmatrix}
\end{align*}

\begin{thm}       
\label{thm52}
The eigenvalues of $\gammaunderbar\ctimes\cappunderbar _n$ are $-\sqrt{2j^0}$ and $\sqrt{2j^0}$, $j=0,1,\ldots ,4^n-1$. These eigenvalues have multiplicity two and their corresponding eigenvectors are $u_k\yhat _{n,j}$, $K=1,2,3,4$.
\end{thm}
\begin{proof}
By direct verification we have
\begin{align*}
\gammaunderbar\ctimes\cappunderbar _nu_1\yhat _{n,j}&=
\begin{bmatrix}\noalign{\smallskip}\paren{\sqrt{2}-1}j^0-j^3+\paren{i\sqrt{j^1}+\sqrt{j^2}}\paren{\sqrt{j^1}-\sqrt{j^2}}\\\noalign{\smallskip}
i\paren{i\sqrt{j^1}-\sqrt{j^2}}\sqrt{j^3}-i\sqrt{j^3}\paren{i\sqrt{j^1}-\sqrt{j^2}}\\\noalign{\smallskip}
\paren{\sqrt{2}-1}\sqrt{j^0}\paren{-i\sqrt{j^3}}-i\sqrt{j^3j^0}\\\noalign{\smallskip}
\paren{\sqrt{2}-1}\sqrt{j^0}\paren{-i\sqrt{j^1}+\sqrt{j^2}}-\sqrt{j^0}\paren{i\sqrt{j^1}-\sqrt{j^2}}\\
\noalign{\smallskip}\end{bmatrix}\,\yhat _{n,j}\\\noalign{\smallskip}
&=\begin{bmatrix}\noalign{\smallskip}\paren{\sqrt{2}-2}j^0\\\noalign{\smallskip}0\\\noalign{\smallskip}
-i\sqrt{2j^0j^3}\\\noalign{\smallskip}
\sqrt{2j^0}\paren{-i\sqrt{j^1}+\sqrt{j^2}}\\\noalign{\smallskip}\end{bmatrix}\,\yhat _{n,j}=-\sqrt{2j^0}\ u_1\yhat _{n,j}
\end{align*}
\begin{align*}
\gammaunderbar&\ctimes\cappunderbar _nu_2\yhat _{n,j}=
\begin{bmatrix}\noalign{\smallskip}i\sqrt{j^3}\paren{i\sqrt{j^1}+\sqrt{j^2}}-i\sqrt{j^3}\paren{i\sqrt{j^1}+\sqrt{j^2}}\\\noalign{\smallskip}
j^0\paren{\sqrt{2}-1}-j^3+\paren{i\sqrt{j^1}+\sqrt{j^2}}\paren{i\sqrt{j^1}-\sqrt{j^2}}\\\noalign{\smallskip}
\paren{\sqrt{2}-1}\sqrt{j^0}\paren{-i\sqrt{j^1}-\sqrt{j^2}}-\sqrt{j^0}\paren{i\sqrt{j^1}+\sqrt{j^2}}\\\noalign{\smallskip}
i\sqrt{j^3}\paren{\sqrt{2}-1}\sqrt{j^0}+i\sqrt{j^3j^0}\\\noalign{\smallskip}\end{bmatrix}\,\yhat _{n,j}\\\noalign{\smallskip}
&=\begin{bmatrix}\noalign{\smallskip}0\\\noalign{\smallskip}j^0\paren{\sqrt{2}-1}-j^3-j^1-j^2\\\noalign{\smallskip}
-\sqrt{2j^0}\paren{i\sqrt{j^1}+\sqrt{j^2}}\\\noalign{\smallskip}i\sqrt{2j^3j^0}\\\noalign{\smallskip}\end{bmatrix}\,\yhat _{n,j}\ 
=\begin{bmatrix}\noalign{\smallskip}0\\\noalign{\smallskip}\paren{\sqrt{2}-2}-j^0\\\noalign{\smallskip}
-\sqrt{2j^0}\paren{i\sqrt{j^1}+\sqrt{j^2}}\\\noalign{\smallskip}-\sqrt{2j^0}\paren{-i\sqrt{j^3}}\\\noalign{\smallskip}\end{bmatrix}\,\yhat _{n,j}
\\\noalign{\smallskip}&=-\sqrt{2j^0}\ u_2\yhat _{n,j}
\end{align*}
\begin{align*}
\gammaunderbar\ctimes\cappunderbar _nu_3\yhat _{n,j}&=
\begin{bmatrix}\noalign{\smallskip}\sqrt{j^0}\paren{i\sqrt{j^1}+\sqrt{j^2}}+\paren{\sqrt{2}-1}\sqrt{j^0}\paren{i\sqrt{j^1}+\sqrt{j^2}}\\\noalign{\smallskip}-i\sqrt{j^3j^0}-\paren{\sqrt{2}-1}\sqrt{j^3j^0}\\\noalign{\smallskip}
-i\sqrt{j^3}\paren{i\sqrt{j^1}+\sqrt{j^2}}-i\sqrt{j^3}\paren{-i\sqrt{j^1}+\sqrt{j^2}}\\\noalign{\smallskip}
\paren{i\sqrt{j^1}+\sqrt{j^2}}\paren{-i\sqrt{j^1}+\sqrt{j^2}}+j^3-\paren{\sqrt{2}-1}j^0\\\noalign{\smallskip}\end{bmatrix}\,\yhat _{n,j}\\\noalign{\smallskip}
&=\begin{bmatrix}\noalign{\smallskip}\sqrt{2j^0}\paren{i\sqrt{j^1}+\sqrt{j^2}}\\\noalign{\smallskip}\sqrt{2j^0}\paren{-i\sqrt{j^3}}\\\noalign{\smallskip}0\\\noalign{\smallskip}j^1+j^2+j^3-\paren{\sqrt{2}-1}j^0\\\noalign{\smallskip}\end{bmatrix}\,\yhat _{n,j}\ 
=\sqrt{2j^0}\,u_3\yhat _{n,j}\end{align*}
\begin{align*}
\gammaunderbar\ctimes\cappunderbar _nu_4\yhat _{n,j}&=
\begin{bmatrix}\noalign{\smallskip}i\sqrt{j^3j^0}+i\sqrt{j^3}\paren{\sqrt{2}-1}\sqrt{j^0}\\\noalign{\smallskip}
\sqrt{j^0}\paren{i\sqrt{j^1}-\sqrt{j^2}}+\paren{i\sqrt{j^1}-\sqrt{j^2}}\paren{\sqrt{2}-1}\sqrt{j^0}\\\noalign{\smallskip}
j^3-\paren{i\sqrt{j^1}-\sqrt{j^2}}\paren{i\sqrt{j^1}+\sqrt{j^2}}-\paren{\sqrt{2}-1}j^0\\\noalign{\smallskip}
i\sqrt{j^3}\paren{-i\sqrt{j^1}+\sqrt{j^2}}+i\sqrt{j^3}\paren{i\sqrt{j^1}-\sqrt{j^2}}\\\noalign{\smallskip}\end{bmatrix}\,\yhat _{n,j}\\\noalign{\smallskip}
&=\begin{bmatrix}\noalign{\smallskip}\sqrt{2j^0}\paren{i\sqrt{j^3}}\\\noalign{\smallskip}
\sqrt{2j^0}\paren{i\sqrt{j^1}+\sqrt{j^2}}\\\noalign{\smallskip}
j^3+j^1+j^2-\paren{\sqrt{2}-1}j^0\\\noalign{\smallskip}0\\\noalign{\smallskip}\end{bmatrix}\,\yhat _{n,j}=\sqrt{2j^0}\,u_4\yhat _{n,j}\qedhere
\end{align*}
\end{proof}

Another way to express the results of Theorem~\ref{thm52} is the following. We can write $\gamma _0=\sigma _0\otimes\sigma _3$ and
$\gamma _k=\sigma _k\otimes\sigma _2$, $k=1.2.3$. We also have that
\begin{equation*}
\cappunderbar _n\ctimes\gammaunderbar =P_n^0\otimes\gamma _0-P_n^1\otimes\gamma _1
  -P_n^2\otimes\gamma _2-P_n^3\otimes\gamma _3
\end{equation*}
On $H_n\otimes\complex ^4$. For $f\in H_n$, $v\in\complex ^4$ with $v=(v_0,v_1,v_2,v_3)$ we have that
\begin{align*}
\cappunderbar _n\ctimes\gammaunderbar f(x_{n,j})\otimes v&=P_n^0f(x_{n,j})\gamma _0v-\sum _{k=1}^3P_n^kf(x_{n,j})\gamma _kv\\
&=f(x_{n,j})\paren{\sqrt{j^0}\ \gamma _0-\sqrt{j^1}\ \gamma _1-\sqrt{j^2}\ \gamma _2-\sqrt{j^3}\ \gamma _3}v\\\noalign{\smallskip}
&\hskip -6pc =f(x_{n,j})\begin{bmatrix}\noalign{\smallskip}\sqrt{j^0}&0&i\sqrt{j^3}&i\sqrt{j^1}+\sqrt{j^2}\\\noalign{\smallskip}
0&\sqrt{j^0}&i\sqrt{j^1}+\sqrt{j^2}&-i\sqrt{j^3}\\\noalign{\smallskip}
-i\sqrt{j^3}&-i\sqrt{j^1}+\sqrt{j^2}&-\sqrt{j^0}&0\\\noalign{\smallskip}
-i\sqrt{j^1}+\sqrt{j^2}&i\sqrt{j^3}&0&-\sqrt{j^0}\\\noalign{\smallskip}\end{bmatrix}\,\yhat _{n,j}\\\noalign{\smallskip}
\end{align*}
To find the eigenvalues and eigenvectors of $\cappunderbar _n\ctimes\gammaunderbar$, let $f=\yhat _{n,j}$. Then
\begin{equation*}
\cappunderbar\ctimes\gammaunderbar\,\yhat _{n,j}\otimes v=\yhat _{n,j}\otimes\junderbartilde\ctimes\gammaunderbar\,v
\end{equation*}
where $\junderbartilde =\paren{\sqrt{j^0},\sqrt{j^1},\sqrt{j^2},\sqrt{j^3}}$. Hence, we require solutions to
$\junderbartilde\ctimes\gammaunderbar v=\lambda v$. We conclude that the eigenvalues are $-\sqrt{2j^0}$ and $\sqrt{2j^0}$ with corresponding eigenvectors $\yhat _{n,j}\otimes u_k$, $k=1,2,3,4$.

\section{Multiverse Probabilities} 
This section proposes specific values for the coupling constants. These values result in multiverse probabilities that predict a preponderance of pulsating universes. Although these universes are not exactly like the ideal pulsating universes considered in Section~3, they have similar properties. Whether these coupling constants provide an approximation to general relativity will be left for later studies.

In the most elementary case, the coupling constants $c_{n,j}^k$ would be independent of $n,j$ so there would only be four constants $c^k$, $k=0,1,2,3$. Probably the simplest nontrivial values for these constants are \cite{gud142,gud15}:
\begin{equation*}
c^0=\cos ^2\theta e^{2i\theta},\quad c^1=c^2=-i\cos\theta\sin\theta e^{2i\theta},\quad c^3=-\sin ^2\theta e^{2i\theta}
\end{equation*}
We can replace $e^{2i\theta}$ by $1$ because this is an overall phase factor that does not affect probabilities. We can then write
\begin{equation}         
\label{eq61}
c_{n,j}^0=\cos ^2\theta,\quad c_{n,j}^1=c_{n,j}^2=-i\cos\theta\sin\theta,\quad c_{n,j}^3=-\sin ^2\theta
\end{equation}
where $\theta\in (-\pi ,\pi )$.

For $j=0,1,\ldots ,4^n-1$ we write $j$ in its quartic representation \eqref{eq51}. Then $x_{n,j}$ is reached by the path
$(j_{n-1},j_{n-2},\ldots ,j_1,j_0)$; that is, the path
\begin{equation*}
(x_{0,0},x_{1,j_{n-1}},x_{2,j_{n-1}4+j_{n-2}},\ldots ,x_{n,j_{n-1}4^{n-1}+j_{n-2}4^{n-2}+\cdots +j_14+j_0})
\end{equation*}
For example $50=3\ctimes 4^2+0\ctimes 4+2$ so we have the quartic representation $50=302$. Hence, $x_{3,50}$ is reached by that path
$(3,0,2)$ which we can write as
\begin{equation*}
(x_{0,0},x_{1,3},x_{2,12},x_{3,50})
\end{equation*}

The amplitude of $x_{n,j}$ becomes
\begin{equation*}
a(x_{n,j})
  =c_{0,0}^{j_{n-1}}c_{1,j_{n-1}}^{j_{n-2}}c_{2,j_{n-1}4+j_{n-2}}^{j_{n-3}}\cdots c_{n-1,j_{n-1}4^{n-2}+j_{n-2}4^{n-3}+\cdots +j_1}^{j_0}
\end{equation*}
For example, $a(x_{3,50})=c_{0,0}^3c_{1,3}^0c_{2,2}^2$. Now let $j(k)$ be the number of $k$s in the quartic representation of $j$ where $k=\brac{0,1,2,3}$. If $c_{n,j}^k$ have the form \eqref{eq61} we have that
\begin{equation*}
a(x_{n,j})=(c_{n,j}^0)^{j(0)}(c_{n,j}^1)^{j(1)}(c_{n,j}^2)^{j(2)}(c_{n,j}^3)^{j(3)}
\end{equation*}
For example, in this case we conclude that
\begin{equation*}
a(x_{3,50})=c_{n,j}^0c_{n,j}^2c_{n,j}^3=i\cos ^3\theta\sin ^3\theta = \tfrac{i}{8}(\sin 2\theta )^3
\end{equation*}

In a very simple model, let us assume that $c_{n,j}^k$ have the form \eqref{eq61} and that $\theta$ is moderately small, say $\theta =1/10$ radians. Then
\begin{align*}
c_{n,j}^0&=\cos ^2(1/10)\approx\paren{1-\tfrac{1}{100}}^2=(99/100)^2\\
c_{n,j}^1&=c_{n,j}^2=-i\cos (1/10)\sin (1/10)\approx -i\paren{1-\tfrac{1}{100}}\ctimes\tfrac{1}{10}\approx -i/10\\
c_{n,j}^3&=-\sin ^2(1/10)\approx -1/100
\end{align*}
In the usual quantum formalism, the probability that the system is at $x_{n,j}$ is
\begin{equation*}
P(x_{n,j})=\ab{a(x_{n,j})}=\ab{c_{n,j}^0}^{2j(0)}\ab{c_{n,j}^1}^{2j(1)}\ab{c_{n,j}^2}^{2j(2)}\ab{c_{n,j}^3}^{2j(3)}
\end{equation*}
Suppose $n$ is large. Then among the $c$-causets $x_{n,j}$, we have that $x_{n,0}$ has the largest probability and $x_{n,4^n-1}$ has essentially zero probability. However, there are a large number of other $c$-causets (paths). In particular, there are many $c$-causets for which $j(0)$ is large, $j(3)=0$ and $j(1)$ and/or $j(2)$ are small but nonzero. Taken together, these $c$-causets have a dominate probability. Such $c$-causets have approximately the geometry of a pulsating universe as considered in Section~3. For example, suppose $n=20$ and $j(0)=17$, $j(1)=1$, $j(2)=2$, $j(3)=0$ with
\begin{equation*}
j=20000200000001000000
\end{equation*}
so that $j=2\ctimes 4^{19}+2\ctimes 4^{14}+4^6$. Then $x_{n,j}$ has the shell sequence $(1,10,17,13)$.

To obtain a more realistic model, we propose that $\theta _{n,j}$ is a function of $n$ that provides a more periodic behavior. This would result in a closer approximation to a pulsating universe as considered in Section~3. In particular, we suggest that $c_{n,j}^k$ have the form \eqref{eq61} with
\begin{equation*}
\theta _n=\frac{1}{\ln (n+1)}\,\cos\sqbrac{\ln (n+1)}
\end{equation*}
We would then have
\begin{align*}
c_{n,j}^0&=\cos ^2(\theta _n)=\cos ^2\brac{\frac{1}{\ln (n+1)}\,\cos\sqbrac{\ln(n+1)}}\\
c_{n,j}^1&=c_{n,j}^2=-\tfrac{i}{2}\,\sin\brac{\frac{2}{\ln (n+1)}\,\cos\sqbrac{\ln (n+1)}}\\
c_{n,j}^3&=c_{n,j}^0-1
\end{align*}
Upon graphing these functions one will see a kind of periodic behavior with increasing periods.

\end{document}